\documentclass[12pt]{article}
\usepackage{amssymb}
\usepackage[latin1]{inputenc}
\usepackage[T1]{fontenc}
\usepackage{amsthm}
\usepackage{amsmath}
\usepackage{eso-pic}
\usepackage{graphicx}
\usepackage{color}
\usepackage{comment}
\newtheorem{thm}{Theorem}
\theoremstyle{definition}
\newtheorem{defn}{Definition}
\newtheorem{lem}{Lemma}

\newtheorem{prop}{Proposition}
\begin{document}

\begin{comment}
\newcommand\BackgroundPicture{%
   \put(0,0){%
     \parbox[b][\paperheight]{\paperwidth}{%
       \vfill
       \centering
       \includegraphics[width=\paperwidth,height=\paperheight, keepaspectratio]{image}%
       \vfill
     }}}

  \AddToShipoutPicture{\BackgroundPicture}

\end{comment}
% The picture is centered on the page background

\date{}

%07 June2012
%Non-Linear--->delete in the title

\title{On the  Conjecture on  APN  Functions}

\author{Moises  Delgado\thanks{Department of Mathematics,
University of Puerto Rico (UPR), Rio Piedras Campus, San Juan,  PR 00936,
 USA. moises.delgado@upr.edu.}\\
Heeralal  Janwa \thanks{Department of Mathematics (UPR), University of
Puerto Rico, Rio Piedras Campus, San Juan,  PR 00936, USA.
heeralal.janwa@upr.edu}} \maketitle

\begin{centering}
\textbf{Abstract} \\
\end{centering}
An almost perfect nonlinear (APN) function  (necessarily  a polynomial function)  on a finite field   $\mathbb{F}$  is called exceptional APN,   if it is also APN on infinitely many extensions
of  $\mathbb{F}$. \textcolor{black}{ In this article we consider the most studied case of $\mathbb{F}=\mathbb{F}_{2^n}$.}
 A conjecture of Janwa-Wilson and McGuire-Janwa-Wilson (1993/1996)\textcolor{black}{,}   settled  in 2011\textcolor{black}{,}  was that the only
\textcolor{black}{exceptional  monomial} APN \textcolor{black}{functions} are \textcolor{black}{the} monomials $x^n$\textcolor{black}{,} where $n=2^i+1$  or
$n={2^{2i}-2^i+1} $ (the Gold or the Kasami exponents respectively).
\textcolor{black}{A   subsequent conjecture states that any  exceptional  APN function is one of the monomials just described.}
\textcolor{black}{One of our \textcolor{black}{result} \textcolor{black}{is that}  all functions  of the form \textcolor{black}{$f(x)=x^{2^k+1}+h(x)$ (for any odd degree  $h(x)$, with a mild condition in few cases}),   are  not exceptional APN, extending  substantially several recent results towards the resolution of the  stated conjecture.}

%An APN function $f:\mathbb{F}_{2^n}\rightarrow \mathbb{F}_{2^n}$ is
%called exceptional APN if it is also APN on infinitely many
%extensions of $\mathbb{F}_{2^n}$. It was conjectured that the only
%exceptional APN functions are the Gold and Kasami-Welch functions.
%It is proved that a polynomial function of odd degree is not
%exceptional APN provided the degree is not a Gold number $(2^k+1)$
%or a Kasami-Welch number $(2^{2k}-2^k+1)$ [1]. Recent results show
%that for the Gold degree case, some infinite families of polynomials
%can not be exceptional APN.

%In this work we extend substantially
%this recent results by providing new infinite families of Gold
%degree polynomials which are not exceptional APN. As an example we
%prove that a polynomial of the form $f(x)=x^{2^k+1}+h(x)$, where
%$\deg(h)=3 (mod 4)< 2^k+1$, is not exceptional APN.

%In this work we extend substantially
%this recent results by providing new infinite families of Gold
%degree polynomials which are not exceptional APN. As an example

%have at most 2 solutions.
%f(x+a)-f(x)=b

%%%Q in L?

\section{Introduction}
%\textcolor{black}{We use the color to indicate the changes}.

\begin{defn}\label{def1}
\textcolor{black}{Let} $L=\mathbb{F}_q$, with $q=p^n$ for some
positive \textcolor{black}{integer} $n$. A function $f:L\rightarrow L$ is said to be
\textbf{almost perfect nonlinear }(APN) on $L$ if for all $a,b \in
L$, $a \neq 0$, the equation
\begin{equation}
f(x+a)-f(x)=b
\end{equation}
has at most 2 solutions.
\end{defn}

Equivalently, $f$ \textcolor{black} is APN if \textcolor{black}{the cardinality of}
the set $\{f(x+a)\textcolor{black}{-}f(x):x \in L \}$ \textcolor{black}{is} at least
$2^{n-1}$ for each $a\in L^{\ast}$. These  \textcolor{black}{kind}
of functions \textcolor{black}{ are} important in applications to cryptography,
where they are  used as S-Boxes\textcolor{black}{,} because they
are  resistant to differential cryptanalytic  attacks. The best
known examples of APN functions are the Gold
\textcolor{black}{function} $f(x)=x^{2^k+1}$\textcolor{black}{,}  and the Kasami-Welch
function\textcolor{black}{ $f(x)=x^{2^{2k}-2^k+1}$,} that  are APN on any field
$\mathbb{F}_{2^n}$\textcolor{black}{,} where $k$
\textcolor{black}{and} $n$ are relatively
\textcolor{black}{prime}. The \textcolor{black}{Welch} function $f(x)=x^{2^r+3}$
\textcolor{black}{is also APN on $\mathbb{F}_{2^n}$, where $n=2r+1$.}\\
The APN property is invariant under some transformations of
functions.\\
A function $f:L \rightarrow L$    is linear if and only if $f$ is a
linearized polynomial over $L$, that is,
$$\sum_{i=0}^{n-1}c_ix^{\textcolor{black}{p}^i}, \,\,\,\,\,\, c_i\in L.$$
The sum of a linear function and a constant is called an affine
function.\\
 Two functions are called {\it extended affine equivalent}
\textcolor{black}{(EA equivalence)},  $f \equiv g$  \textcolor{black}{
(EA}),  if
$f=A_1\circ g \circ A_2+A$\textcolor{black}{,} where $A_1$ \textcolor{black}{and} $A_2$ are \textcolor{black}{linear maps
and $A$ is a constant function}. A second equivalence is the CCZ
equivalence, $f \equiv g (CCZ)$ if the graph of $f$ can be obtained
from the graph of $g$ by an affine permutation. EA equivalence is a
particular case of CCZ equivalence; two CCZ equivalent functions
preserve the APN property (for more details see [5]). In
general,  CCZ equivalence is very difficult to establish.

%Mishra: proving--->delte
%The research on the APN functions and its applications became very
%important for the mathematicians in the last years.
%Mishra--->for--->form
%Mishra the the---->the
Until 2006, the
list of known affine inequivalent APN functions on $L=GF(2^n)$ was
rather short; the list consisted only of monomial functions of the
form  $f(x)=x^t$\textcolor{black}{,} for some positive \textcolor{black}{integer} $t$. In February 2006,
\textcolor{black}{Y. Edel}, \textcolor{black}{G. Kyureghyan} and \textcolor{black}{A. Pott} [6]  established (by an exhaustive search) the  first
example of an APN function not equivalent to any of
the known monomial APN  functions. Their  example is
$$x^3+ux^{36} \in GF(2^{10})[x],$$
where $u\in wGF(2^5)^* \cup w^2GF(2^5)^*$ and $w$ has order 3, is APN on $GF(2^{10})$.
Since then,   several new  infinite families of polynomial APN functions have been discovered.

\section{Exceptional APN functions}\label{Exceptional-APN-functions}
\textcolor{black}{
In this section we discuss the main conjecture on
exceptional APN functions. An almost perfect nonlinear (APN) function  (necessarily  a polynomial function)  on a finite field   $\mathbb{F}$  is called exceptional APN,   if it is also APN on infinitely many extensions
of  $\mathbb{F}$.  In this article we consider the most studied case of $\mathbb{F}=\mathbb{F}_{2^n}$.
 A conjecture of Janwa-Wilson and McGuire-Janwa-Wilson (1993/1996),    settled  in 2011,  was that the only
exceptional  monomial APN functions are the monomials $x^n$, where $n=2^i+1$  or
$n={2^{2i}-2^i+1} $ (the Gold or the Kasami exponents respectively). 
The  Welch functions  $f(x)=x^{2^r+3}$  are known to  be APN on $\mathbb{F}_{2^n}$ for $n=2r+1$, but they are  not exceptional APN functions.
The same is the case for the new class of  APN functions discovered  by Edel et. in  \cite{NEW}, discussed earlier.
Since then several  have been proved (see section \ref{sec:Recent-results},  establishing that many infinite classes of functions are not exceptional APN functions (see \cite{AMR}, \cite{BM}, \cite{NEW2}, \cite{J}, \cite{R}, \cite{R2}, and others.)
 Based on these results, an open conjecture  states that any  exceptional  APN function is one of the monomials just described.
}

\textcolor{black}{
Our theorems \ref{thm:3mod4} and \ref{thm:1mod4} in  section \ref{sec:Our-main-results} demonstrate that  all functions  of the form $f(x)=x^{2^k+1}+h(x)$ (for any odd degree  $h(x)$, with a mild condition in few cases),   are  not exceptional APN, extending  substantially several recent results towards the resolution of the  stated conjecture.}  

\begin{defn}\label{def2}
Let $L=\mathbb{F}_q$, \textcolor{black}{with} $q=p^n$ for some positive integer
$n$. A function $f:L\rightarrow L$ is called \textbf{exceptional
APN} if $f$ is APN on $L$ and also on
infinitely many extensions of $L$.
\end{defn}
%Mishra: Make---->made
Aubry, McGuire and Rodier  \textcolor{black}{\cite{AMR} } made  the following \textcolor{black}{conjecture}.

\textbf{CONJECTURE:} Up to equivalence, the Gold and Kasami-Welch
functions are the only exceptional APN functions.

\begin{prop}[\textcolor{black}{Rodier \cite{R}}]
\label{prop1}
 Let $L=\mathbb{F}_q$, with $q=2^n$. A function $f:L\rightarrow L$
is APN if and only if the affine surface $X$ \textcolor{black}{with} equation
$$f(x)+f(y)+f(z)+f(x+y+z)=0$$
has all its rational points contained in the surface
$(x+y)(x+z)(y+z)=0.$
\end{prop}

%Mishra: , guarantees--->guarantee
Using this characterization and the \textcolor{black}{ bounds} of Lang-Weil and
Ghorpade-Lachaud that  guarantee many rational points on a surface
for all $n$ sufficiently large, \textcolor{black}{one} can prove the following
theorem.
% a exceptional ---->an exceptional

\begin{thm} [\textcolor{black}{See Rodier \cite{R}}]
\label{thm1}
Let $f:L\rightarrow L$, $L=\mathbb{F}_{2^n}$, a polynomial
function of degree $d$. \textcolor{black}{Suppose \textcolor{black}{the} surface $X$}

$$\frac{f(x)+f(y)+f(z)+f(x+y+z)}{(x+y)(x+z)(y+z)}=0$$
is absolutely irreducible (or has an absolutely irreducible
component over $L$), then $f$ is not an exceptional APN function.
\end{thm}

\textcolor{black}{From now on, we let}
\begin{equation}
\phi(x,y,z)=\frac{f(x)+f(y)+f(z)+f(x+y+z)}{(x+y)(x+z)(y+z)}
\end{equation}

\begin{equation}
 \phi_j(x,y,z)=\frac{x^j+y^j+z^j+(x+y+z)^j}{(x+y)(x+z)(y+z)}
\end{equation}

%BM: polynomials---->polynomial , and comma  before f

\section{Recent results}\label{sec:Recent-results}
Aubry, McGuire and Rodier  \textcolor{black}{\cite{AMR} } proved the following  result that  provides infinite families of polynomial functions that are
not exceptional APN. As before,  let $L=\mathbb{F}_q$, with $q=2^n$ \textcolor{black}{and}
$f:L\rightarrow L$ a polynomial function.

\begin{thm}\label{thm2}
\textcolor{black}{If the degree of  $f$ is odd
and not a Gold or a Kasami-Welch number\textcolor{black}{,} 
then $f$ is not APN for all sufficiently large extensions of $L$}.

\end{thm}

For the even degree case, they proved the following:

\begin{thm} \label{thm3}
\textcolor{black}{If the degree of  $f$ is $2e$ with
$e$ odd, and if $f$ contains an odd degree term,  
then $f$ is not
APN for all sufficiently large extensions of $L$}.

\end{thm}

\begin{thm} \label{thm4}
\textcolor{black}{If the degree of $f$ is $4e$ with $e\geq 7$ and $e\equiv 3 \pmod 4$, then $f$ is not
APN for all sufficiently large extensions of $L$}.
\end{thm}
They also found results for Gold degree polynomials.

\begin{thm} \label{thm:AMR}
Suppose $f(x)=x^{2^k+1}+g(x)$\textcolor{black}{,} where $\deg(g)\leq
2^{k-1}+1$. Let $g(x)=\sum_{j=0}^{2^{k-1}+1}a_jx^j$. Suppose that
there exists a nonzero coefficient $a_j$ of $g$ such that
$\phi_j(x,y,z)$ is absolutely irreducible. Then $\phi(x,y,z)$ is
absolutely irreducible and $f$ is not exceptional APN.
\end{thm}

In this last theorem the authors showed that the bound $\deg(g)\leq
2^{k-1}+1$ is best possible, in the sense that there is an example
with $\deg(g)= 2^{k-1}+2$ where $\phi(x,y,z)$ is not absolutely
irreducible [12].
%BM: provides--->provide
\section{\textcolor{black}{Large Classes of} Gold degree families that are not exceptional APN}
In this section we will state and prove \textcolor{black}{several} results that  provide
families of Gold degree polynomials that  are not exceptional APN. \textcolor{black}{We will accomplish this by proving that  the surface $\phi(x,y,z)=0$ related to the
polynomial $f(x)$ is absolutely irreducible. }

\textcolor{black}{One of our \textcolor{black}{result} \textcolor{black}{is that}  all functions  of the form \textcolor{black}{$f(x)=x^{2^k+1}+h(x)$ (for any odd degree  $h(x)$, with a mild condition in few cases}),   are  not exceptional APN, extending  substantially several recent results towards the resolution of the  stated conjecture.}

\textcolor{black}{We begin with the} following facts,
due to Janwa and Wilson  \textcolor{black}{ \cite{JMW1}},
about the Gold, Kasami-Welch and Welch \textcolor{black}{ functions.}

If $f(x)=x^{2^k+1}$ is a Gold function, then
\begin{equation}
\phi(x,y,z)=\prod_{\alpha \in F_{2^k}-F_2}(x+\alpha y+(\alpha+1)z)\textcolor{black}{.}
\end{equation}
If $f(x)=x^{2^{2k}-2^k+1}$ \textcolor{black}{is} a Kasami Welch function, then
\begin{equation}
\phi(x,y,z)=\prod_{\alpha \in F_{2^k}-F_2}P_{\alpha}(x,y,z)\textcolor{black}{,}
\end{equation}
where $P_{\alpha}(x,y,z)$ is absolutely irreducible of degree
$2^k+1$ over $GF(2^k)$.\\
If $f(x)=x^{2^k+3}$ is a Welch function, then $\phi(x,y,z)$ is
absolutely irreducible for $k>1$.\\

\textcolor{black}{We will frequently use the following lemma.}

\begin{lem} \label{lem1}
\textcolor{black}{For} an integer  $k > 1$, let $l=2^k+1$, $m=2^{2k}-2^k+1$ \textcolor{black}{and}
$n=2^k+3$ \textcolor{black}{be} a Gold, Kasami-Welch and Welch numbers respectively\textcolor{black}{.}
\textcolor{black}{Then} $(\phi_l,\phi_m)=1$, $(\phi_l,\phi_n)=1$ \textcolor{black}{and} $(\phi_m,\phi_n)=1$. Also:\\
a) If $l_1=2^{k_1}+1$ \textcolor{black}{and} $l_2=2^{k_2}+1$ are different Gold numbers \textcolor{black}{such that} $(k_1,k_2)=1,$
\textcolor{black}{then $(\phi_{l_1},\phi_{l_2})=1$.}\\
b)$(\phi_{m_1},\phi_{m_2})=1$ for different Kasami-Welch numbers
$m_1$ \textcolor{black}{and} $m_2$.\\
c)$(\phi_{n_1},\phi_{n_2})=1$ for different Welch numbers $n_1$ \textcolor{black}{and} $n_2$.
\end{lem}

\begin{proof}
\textcolor{black}{
The proof of this lemma follows directly from (4), (5) and the fact
that $\phi_n$, for a Welch number $n>5$, is absolutely irreducible.
}
\end{proof}

\subsection{\textcolor{black}{Overcoming the obstacle}}\label{sec:Overcoming-the-obstable}

\textcolor{black}{Theorem \ref{thm:AMR} shows that $f(x)=x^{2^k+1}+g(x)$ with $\deg(g)= 2^{k-1}+1$ provides an obstacle to absolute irreducibility.}
Now we will show that there are cases where it is possible to increase the degree of $g(x)$ in order
to obtain new non-exceptional APN functions.\\
From now on, let $L=\mathbb{F}_{2^n}$, $f:L \rightarrow L$ and
$\phi(x,y,z)$, $\phi_j(x,y,z)$ as in (2) and (3). (The polynomial
$f$ can be considered not containing a constant term nor terms of
degree power of two, since as we commented APN property is invariant
under affine maps)\textcolor{black}{.}

\begin{thm} \label{thm:obstacle}
For $k \geq 2$ \textcolor{black}{and} $\alpha \neq 0$\textcolor{black}{,} let
$f(x)=x^{2^k+1}+\alpha x^{2^{k-1}+3}+h(x) \in L[x]$\textcolor{black}{,} where
$h(x)=\sum_{j=0}^{2^{k-1}+1}a_jx^j$ and satisfy one of the following
conditions:\\
a) $a_5=0$.\\
b) There is a non zero $a_j\phi_j$ for some $j\neq 5$.\\
Then $\phi(x,y,z)$ is absolutely irreducible.
\end{thm}
%lose--->loss
%used in ---->used in the
\begin{proof}
If $k=2$, \textcolor{black}{then}  $f(x)=x^5+b_3x^3$, \textcolor{black} {and it  is known that it is not exceptional}
APN [2].\\
If $k=3$, \textcolor{black}{then} $f(x)=x^9+a_7x^7+a_5x^5+a_3x^3$,  \textcolor{black} {and it  is known that it is not exceptional}
 APN [12].\\
Let $k>3$ and let \textcolor{black}{$\phi(x,y,z)=0$} \textcolor{black}{be} the surface \textcolor{black}{corresponding} to $f(x)$.
Suppose \textcolor{black}{that} $\phi$ is not absolutely irreducible, then $\phi(x,y,z)=P(x,y,z)Q(x,y,z)$, where $P$ \textcolor{black}{and} $Q$ are \textcolor{black}{non-constant} polynomials. \textcolor{black}{We write} $P$ \textcolor{black}{and} $Q$ \textcolor{black}{as sums} of homogeneous terms:
\begin{equation}
\sum_{j=3}^{2^{k}+1}
a_j\phi_j(x,y,z)=(P_s+P_{s-1}+...+P_0)(Q_t+Q_{t-1}+...+Q_0)\textcolor{black}{,}
\end{equation}
where $P_j$ \textcolor{black}{and} $Q_j$ are zero or homogeneous of degree $j$,
$s+t=2^{k}-2$. Assuming\textcolor{black}{,} without loss of generality\textcolor{black}{,} that $s\geq t$\textcolor{black}{.}
\textcolor{black}{Then,} $$2^{k}-2>s\geq \frac{2^{k}-2}{2} \geq t>0\textcolor{black}{.}$$
\textcolor{black}{In (6),}
\begin{equation}
P_sQ_t=\phi_{2^k+1}\textcolor{black}{,}
\end{equation}
since $\phi_{2^k+1}$ is equal to the product of
different linear factors, $P_s$ and $Q_t$ are relatively \textcolor{black}{prime}.\\
By the assumed degree of $h(x)$, the homogeneous \textcolor{black}{terms} of degree $r$,
for $2^{k-1}<r<2^{k}-2$, are equal to zero. Then\textcolor{black}{,} equating the terms
of degree $s+t-1$ gives $P_sQ_{t-1}+P_{s-1}Q_t=0$. Hence\textcolor{black}{,} we \textcolor{black}{have} $P_s$ divides $P_{s-1}Q_t$ and this implies that $P_s$ divides
$P_{s-1}$, since $P_s$ and $Q_t$ are relatively \textcolor{black}{prime}. We conclude
that $P_{s-1}=0$ as the degree of $P_{s-1}$ is less than the degree
of $P_{s}$. Then\textcolor{black}{,} we also \textcolor{black}{have} $Q_{t-1}=0$ as $P_s\neq 0$.\\
Similarly, equating the terms of degree $s+t-2, s+t-3,...,s+2$ we get:
\textcolor{black}{$$P_{s-2}=Q_{t-2}=0,$$ $$P_{s-3}=Q_{t-3}=0,$$ $$...$$ $$P_{s-(t-2)}=Q_2=0.$$}\\
The \textcolor{black}{(simplified) equations of degree $s+1$, $s$, $s-1$ and $s-2$,} respectively,
are:
$$P_sQ_1+P_{s-(t-1)}Q_t=a_{s+4}\phi_{s+4}$$
$$P_sQ_0+P_{s-t}Q_t=a_{s+3}\phi_{s+3}$$
$$P_{s-1}Q_0+P_{s-(t+1)}Q_t=a_{s+3}\phi_{s+2}$$
$$P_{s-2}Q_0+P_{s-(t+2)}Q_t=a_{s+3}\phi_{s+1}$$
We  consider two cases:
\\
\textbf{FIRST CASE}: $s > 2^{k-1}$\textcolor{black}{.}\\
For this case\textcolor{black}{,} $t<2^{k-1}-2$. From the equations of degree $s+1$ \textcolor{black}{and} $s$ we have:\\
$P_sQ_1+P_{s-(t-1)}Q_t=0$, then $P_{s-(t-1)}=Q_1=0$;
$P_sQ_0+P_{s-t}Q_t=0$, then $P_{s-t}=Q_0=0$ (using the same argument
we used in the equation $s+t-1$). Then\textcolor{black}{,} $Q=Q_t$ is homogeneous of degree
$t$ and by (7) and (6) there exist some $\alpha \in F_{2^k}-F_2$
such that $x+\alpha y+(1+\alpha)z$ divides both $\phi_{2^k+1}$ \textcolor{black}{and} $\phi_{2^{k-1}+3}$\textcolor{black}{, contradicting lemma \textcolor{black}{\ref{lem1}}.}
\\
\textbf{SECOND CASE}: $2^{k-1} \geq s \geq 2^{k-1}-1$\textcolor{black}{.}\\
For this case\textcolor{black}{,} $2^{k-1}-2 \leq t\leq 2^{k-1}-1$. Then\textcolor{black}{,} we have the subcases:\\
SUBCASE 1: ($s=t=2^{k-1}-1$)\\
We already \textcolor{black}{have}
$P_{s-1}=Q_{t-1}=P_{s-2}=Q_{t-2}=...=P_2=Q_2=0$.\\ The \textcolor{black}{equation of degree $s+1$ is}
\begin{equation}
P_sQ_1+P_1Q_t=\alpha\phi_{s+4}\textcolor{black}{.}
\end{equation}
The equation of degree $s$ is $P_sQ_0+P_0Q_s=0$. Using the argument
\textcolor{black}{from} the first case we \textcolor{black}{get} $P_0=Q_0=0$. For $r<s$, $r\neq5$, the
equation of degree $\leq r$ is of the form $0=a_r\phi_r$ (since all
$P_i, Q_i$ are zero except $P_s,Q_t$ and possibly $P_1$ \textcolor{black}{and} $Q_1$). Then\textcolor{black}{,}
if for some $j\neq 5$, $a_j\phi_j\neq 0$\textcolor{black}{,} we are done. If $a_5 =0$,
then the equation of degree \textcolor{black}{two} is $P_1Q_1=0$, \textcolor{black}{so} one of them is
equal to zero. If $P_1=0$, then the equation (8) becomes
$P_sQ_1=\alpha\phi_{s+4}$\textcolor{black}{,} \textcolor{black}{contradiction} since
$\phi_{s+4}=\phi_{2^{k-1}+3}$ is absolutely
irreducible. The case $Q_1=0$ is similar.\\
SUBCASE 2: ($s=2^{k-1}, t=2^{k-1}-2$)\\
We already \textcolor{black}{have}
$P_{s-1}=Q_{t-1}=P_{s-2}=Q_{t-2}=...=P_4=Q_2=0$.\\ The \textcolor{black}{equation of degree $s+1$ is} $P_sQ_1+P_3Q_t=0$\textcolor{black}{, so} $P_3=Q_1=0$.\\ The \textcolor{black}{equation of degree $s$ is:}
\begin{equation}
P_sQ_0+P_2Q_t=\alpha\phi_{s+3}\textcolor{black}{.}
\end{equation}
The equation of degree $s-1$ is $P_1Q_t=0$, then $P_1=0$.\\ The equation of degree $s-2$ is $P_0Q_t=a_{s+1}\phi_{s+1}$. If
$a_{s+1}\neq 0$\textcolor{black}{,} then $Q_t$ divides $\phi_{s+1}=\phi_{2^{k-1}+1}$, but $Q_t$ also divides $\phi_{2^k+1}$\textcolor{black}{, and} that  is not possible by lemma \textcolor{black}{\ref{lem1}}.
Then\textcolor{black}{,} $a_{s+1}=0$ and $P_0=0$.\\
For $r<s-2, r \neq 5$, the \textcolor{black}{equations} of degree $\leq r$ have the form
$0=a_r\phi_r$. Then\textcolor{black}{,} if for some $j\neq 5$, $a_j\phi_j\neq 0$\textcolor{black}{,} we are
done. If $a_5 =0$, then the equation of degree \textcolor{black}{two} is $P_2Q_0=0$, \textcolor{black}{so}
one of them is zero. If $Q_0=0$\textcolor{black}{,} then the equation (9) becomes
$P_2Q_t=\alpha\phi_{s+3}$\textcolor{black}{;} \textcolor{black}{contradiction} \textcolor{black}{to}  the
irreducibility of $\phi_{s+3}$. If $P_2=0$, then
$P_sQ_0=\alpha\phi_{s+3}$, but also $P_sQ_t=a_{2^k+1}\phi_{2^k+1}$\textcolor{black}{,}
contradicting lemma \textcolor{black}{\ref{lem1}}.
\end{proof}

Some families \textcolor{black} covered by this theorem:\\
$f(x)=x^{17}+a_{11}x^{11}+h(x)$\textcolor{black}{,} where $a_{11}\neq 0, \deg(h)\leq 9$,
except the case $f(x)=x^{17}+a_{11}x^{11}+a_5x^5$, $a_{11}\neq 0, a_5 \neq 0$\textcolor{black}{;}\\
$f(x)=x^{33}+a_{19}x^{19}+h(x)$\textcolor{black}{,} where $a_{19}\neq 0, \deg(h)\leq
17$, except the case $f(x)=x^{33}+a_{19}x^{19}+a_5x^5$, $a_{19}\neq 0, a_5 \neq 0$\textcolor{black}{.}

As we can notice in this theorem, we interpolate a Welch term
$\alpha x^{2^{k-1}+3}$ between the consecutive Gold terms
$x^{2^k+1}$ and $x^{2^{k-1}+1}$ of the family given in theorem \textcolor{black}{\ref{thm:AMR}}.

Including two \textcolor{black}{or} more terms would imply more cases and subcases to
consider and this is not a good idea. Next, we provide more general
families of Gold degree polynomials that  cannot be exceptional
APN.

\subsection{Hyperplane Sections}

\textcolor{black}{We first prove some results on plane sections.}
  Let us consider the intersection of the surface $\phi(x,y,z)$ with
the plane $y=z$.

\begin{lem} \label{lem2}
Let $\phi_j(x,y,z)$ be as in (3)\textcolor{black}{. Then}\\
a) For $n=2^k+1>3 $, $\phi_n(x,y,y)=(x+y)^{2^k-2}$\textcolor{black}{;}\\
b) For \textcolor{black}{$n\equiv 3 \pmod 4>3 $}, $\phi_n(x,y,y)=R(x,y)$ such that $x+y$ does not divides $R(x,y)$\textcolor{black}{;}\\
c) For \textcolor{black}{$n\equiv 1 \pmod 4>5 $}, $\phi_n(x,y,y)=(x+y)^{2^l-2}S(x,y)$, such
that $x+y$ does not divides $S(x,y)$, where $n=1+2^lm$, $l\geq 2$
and  $m > 1$ is an odd number.
\end{lem}

\begin{proof}
The part a) follows directly from (4). For the part b) we have:
\begin{equation*}
\phi_n(x,y,z)=\frac{x^n+y^n+z^n+(x+y+z)^n}{(x+y)(x+z)(y+z)}
\end{equation*}

\begin{equation*}
\phi_n(x,y,z)=\frac{1}{(x+y)(x+z)}(\frac{y^n+z^n}{y+z}+\frac{x^n+(x+y+z)^n}{y+z})
\end{equation*}
Then
\begin{align*}
\phi_n(x,y,z)=\frac{1}{(x+y)(x+z)}[y^{n-1}+y^{n-2}z+...+yz^{n-2}+z^{n-1}
&\notag\\
&\hspace{-11 cm}
+x^{n-1}+x^{n-2}(x+y+z)+...+x(x+y+z)^{n-2}+(x+y+z)^{n-1}]
\end{align*}

Making $y=z$, we get:
$$\phi_n(x,y,y)=\frac{x^{n-1}+y^{n-1}}{(x+y)^2}$$
Since \textcolor{black}{$n\equiv 3\pmod 4>3$} then $n=3+4m$ for $m\geq 1$. Then
$$\phi_n(x,y,y)=\frac{x^{2+4m}+y^{2+4m}}{(x+y)^2}$$
$$\phi_n(x,y,z)=(\frac{x^{1+2m}+y^{1+2m}}{x+y})^2$$
$$\phi_n(x,y,z)=(x^{2m}+x^{2m-1}y+...+xy^{2m-1}+y^{2m})^2$$
Therefore, $\phi_n(x,y,y)=R(x,y)$. Now, $x+y$ divides $R(x,y)$ if
and only if $x+1$ divides
$x^{2m}+x^{2m-1}+x^{2m-2}+...+x+1$, that is not true.\\
For the part c), \textcolor{black}{with $y=z$ we get similarly}:
$$\phi_n(x,y,y)=\frac{x^{n-1}+y^{n-1}}{(x+y)^2}$$
Since $n=1+2^lm > 5$, for $m>1$ an odd integer:

$$\phi_n(x,y,y)=\frac{x^{2^lm}+y^{2^lm}}{(x+y)^2}$$
$$\phi_n(x,y,z)=\frac{(x^{m}+y^{m})^{2^l}}{(x+y)^2}$$
$$\phi_n(x,y,y)=(x+y)^{2^l-2}(x^{m-1}+x^{m-2}y+...+xy^{m-2}+y^{m-1})^{2^l}$$
Then $\phi_n(x,y,y)=(x+y)^{2^l-2}S(x,y)$. Now, $x+y$ divides
$S(x,y)$ if and only if $x+1$ divides
$x^{m-1}+x^{m-2}+x^{m-3}+...+x+1$ that is not true either.
\end{proof}

\section{\textcolor{black}{Our  main results}}\label{sec:Our-main-results}

\textcolor{black}{Now we prove  our main results  (Theorems \ref{thm:3mod4} and \ref{thm:1mod4} in this section)  establishing  that {$f(x)=x^{2^k+1}+h(x)$ (for any odd degree  $h(x)$, with a mild condition in few cases}),   are  not exceptional APN, extending  substantially several recent results towards the resolution of the  stated conjecture.}
\textcolor{black}{In particular, these two theorems, in a large measure,   extend the result of Aubrey, McGuire and Rodier \cite{AMR} as stated in Theorem \ref{thm:AMR} stated earlier. }

\begin{thm} \label{thm:3mod4}
\textcolor{black}{For} $k \geq 2$, let $f(x)=x^{2^k+1}+h(x) \in L[x]$\textcolor{black}{,}
where \textcolor{black}{$\deg(h)\equiv 3\pmod 4< 2^k+1$}. Then\textcolor{black}{,} $\phi(x,y,z)$ is absolutely irreducible.
\end{thm}

\begin{proof}
As in the previous theorem, for \textcolor{black}{$k=2$ \textcolor{black}{and} $3$}, we have
the \textcolor{black} exceptional
polynomials $f(x)=x^5+b_3x^3$ \textcolor{black}{and} $f(x)=x^9+a_7x^7+a_5x^5+a_3x^3$\textcolor{black}{,} \textcolor{black}{respectively}.\\
Let $k>3$ and let $\phi(x,y,z)$ \textcolor{black}{be} the function related to $f(x)$. \textcolor{black}{As before},
suppose that $\phi(x,y,z)=P(x,y,z)Q(x,y,z)$, $P$ \textcolor{black}{and} $Q$ \textcolor{black}{non-constant}s and writing $P$ \textcolor{black}{and} $Q$ as \textcolor{black}{sums} of homogeneous terms:
\begin{equation}
\sum_{j=3}^{2^{k}+1}
a_j\phi_j(x,y,z)=(P_s+P_{s-1}+...+P_0)(Q_t+Q_{t-1}+...+Q_0)\textcolor{black}{,}
\end{equation}
where $s+t=2^{k}-2$. Assuming that $s\geq t$, then
$2^{k}-2>s\geq \frac{2^{k}-2}{2} \geq t>0$. Let $d=\deg(h)$, $e=2^k+1-d$.\\
In (10), we have:
\begin{equation}
P_sQ_t=\phi_{2^k+1}\textcolor{black}{.}
\end{equation}
Since $\phi_{2^k+1}$ is equal to the product of
different linear factors, we conclude that  $P_s$ and $Q_t$ are relatively \textcolor{black}{prime}.\\
By the assumed degree of $h(x)$, the homogeneous \textcolor{black}{terms} of degree $r$,
for $d-3<r<2^{k}-2$, are equal to zero. Then\textcolor{black}{,} equating the terms of
degree $s+t-1$ gives $P_sQ_{t-1}+P_{s-1}Q_t=0$. As in the previous
theorems, we conclude that $P_{s-1}=0$ and $Q_{t-1}=0$.\\
Similarly, equating the terms of degree $s+t-2, s+t-3,...,s+t-(e-1)=d-2$\textcolor{black}{,} we get: \textcolor{black}{$$P_{s-2}=Q_{t-2}=0,$$ $$P_{s-3}=Q_{t-3}=0,$$ $$...$$ $$P_{s-(e-1)}=Q_{t-(e-1)}=0.$$}
\textcolor{black}{The} \textcolor{black}{equation} of degree $d-3$ is:
\begin{equation}
P_sQ_{t-e}+P_{s-e}Q_t=a_{d}\phi_{d}
\end{equation}
If $Q_{t-e}=0$\textcolor{black}{,} then we \textcolor{black}{have} $P_{s-e}Q_t=a_{d}\phi_{d}$. But\textcolor{black}{,} from (11) we also have $P_sQ_t=\phi_{2^k+1}$ \textcolor{black}{which} is impossible
by lemma \textcolor{black}{\ref{lem1}}. The case $P_{s-e}=0$ is analogous.\\
Then, suppose that $Q_{t-e}\neq 0$ and $P_{s-e}\neq 0$. We
consider the intersection of $\phi(x,y,z)$ \textcolor{black}{with} the plane $y=z$.
Using lemma \textcolor{black}{\ref{lem2}}, the \textcolor{black} {equations} (11) and (12) \textcolor{black}{become}\textcolor{black}{:}
$$P_sQ_t=(x+y)^{2^k-2}$$
$$P_sQ_{t-e}+P_{s-e}Q_t=\phi_d(x,y)\textcolor{black}{.}$$
Therefore\textcolor{black}{,} $(x+y)$ divides both $P_s$ \textcolor{black}{and} $Q_t$\textcolor{black}{,} so it divides
$\phi_d(x,y)$, \textcolor{black}{contradiction} \textcolor{black}{to} part b) of lemma \textcolor{black}{\ref{lem2}}.
Therefore\textcolor{black}{,} $\phi$ is absolutely irreducible.
\end{proof}

This theorem includes theorem \textcolor{black}{\ref{thm:obstacle}} as a particular case. The next
theorem is a version of theorem \textcolor{black}{\ref{thm:3mod4}} for the case
\textcolor{black}{$d\equiv 1 \pmod 4$} with one additional
condition.

\begin{thm} \label{thm:1mod4}
For $k \geq 2$, let $f(x)=x^{2^k+1}+h(x) \in L[x]$
where $d=\deg(h)\textcolor{black}{\equiv 1} \pmod 4< 2^k+1$. If $\phi_{2^k+1}, \phi_d$ are
relatively prime, then $\phi(x,y,z)$ is absolutely irreducible.
\end{thm}

\begin{proof}
As we did before, suppose that $\phi(x,y,z)=P(x,y,z)Q(x,y,z)$, then
\begin{equation}
\sum_{j=3}^{2^{k}+1}
a_j\phi_j(x,y,z)=(P_s+P_{s-1}+...+P_0)(Q_t+Q_{t-1}+...+Q_0)
\end{equation}
where $s+t=2^{k}-2$. Assuming that $s\geq t$, then
$2^{k}-2>s\geq \frac{2^{k}-2}{2} \geq t>0$. Let $e=2^k+1-d$, then $e\geq 2^l$.\\
In (13), we have:
\begin{equation}
P_sQ_t=\phi_{2^k+1}
\end{equation}
since $\phi_{2^k+1}$ is equal to the product of
different linear factors, $P_s$ and $Q_t$ are relatively \textcolor{black}{prime}.\\
By the assumed degree of $h(x)$, the homogeneous
\textcolor{black}{terms} of degree $r$, for $d-3<r<2^{k}-2$, are
equal to zero. Then, as in the previous theorems, equating the terms
of degree $s+t-1, s+t-2, s+t-3,...,d-2$ we get:
$P_{s-1}=Q_{t-1},P_{s-2}=Q_{t-2}=0,
P_{s-3}=Q_{t-3}=0,...,P_{s-(e-1)}=Q_{t-(e-1)}=0$.\\
The equation of degree $d-3$ is:
\begin{equation}
P_sQ_{t-e}+P_{s-e}Q_t=a_{d}\phi_{d}
\end{equation}
We  consider two cases to prove the irreducibility of $\phi$.\\
\textbf{First case} ($s>d-3$)\\
\textcolor{black}{Since} $s+t=2^k-2$, then $t<e$.
\textcolor{black}{The}  equation (15) becomes
$P_{s-e}Q_t=a_{d}\phi_{d}$. But also $P_sQ_t=\phi_{2^k+1}$
\textcolor{black}{contradiction} \textcolor{black}{to} the assumptions of the theorem.\\
\textbf{Second case} ($s \leq d-3$)\\
For this case $t \geq e$. In the equation (15), if $Q_{t-e}=0$ or
$P_{s-e}=0$ we are done, as  in theorem \textcolor{black}{\ref{thm:3mod4}}. In other case, if
$Q_{t-e}\neq 0$ and $P_{s-e}\neq 0$, we consider the intersection of
$\phi(x,y,z)$ \textcolor{black}{with} the plane $y=z$. Then, using
lemma \textcolor{black}{\ref{lem2}}, the \textcolor{black} {equations}  (14) and (15) \textcolor{black}{ become}:\\
$$P_sQ_t=(x+y)^{2^k-2}$$
$$P_sQ_{t-e}+P_{s-e}Q_t=\phi_d(x,y)$$
\textcolor{black}{where} $\phi_d(x,y)=(x+y)^{2^l-2}S(x,y)$ and $x+y$ does not divides
$S(x,y)$. Since $t \geq e \geq 2^l$ then, from the first equation,
$(x+y)^{2^l}$ divides both $P_s, Q_t$ implying that $(x+y)^{2^l}$
divides $\phi_d(x,y)$, that implies that $x+y$ divides $R(x,y)$, a
contradiction.
\end{proof}

\section{Some Applications}

%As a consequence of Theorem \ref{thm:obstacle}, we gave a list of  polynomials are not exceptional APN. 

\textcolor{black}{As a consequence of our  Theorems \ref{thm:obstacle}, \ref{thm:3mod4}, and \ref{thm:1mod4}\textcolor{black}{,} and using results from
Janwa and Wilson \cite{JMW1} and Janwa, McGuire and Wilson \cite{JMW2}, we are able to prove  the following result.}
\begin{thm}
\textcolor{black}{All polynomials of the form $f(x)=x^{65}+h(x)$ are not exceptional APN for all odd degree polynomials $h$.}
 \end{thm}

%\textcolor{black}{We omit the proof.}

\textcolor{black}{We thus extend substantially the classification of all lower degree exceptional APN functions given in literature.}

\section{\textcolor{black}{Open Problems and Future Directions}}

Janwa, McGuire and Wilson \textcolor{black}{ \cite{JMW2} }proved that, for
\textcolor{black}{$t\equiv5 \pmod 8>13$}, if the maximal cyclic
code $B_l$ ($l$ is its length) has no codewords of weight 4, then
$\phi_t(x,y,z)$ is absolutely irreducible. For many values of $l$ it
is possible that $B_l$ has no codewords of weight 4, for example, if
$l$ is a prime congruent to $\pm 3 \pmod 8.$ For more details and
infinite classes, see [9].

\textcolor{black}{Until recently, it was thought that $\phi_d(x,y,z)$ was
absolutely irreducible for the values of \textcolor{black}{$d\equiv 5 \pmod 8$.} F. Hernando
and G. McGuire, with the help of MAGMA, found that the polynomial
$g_{205}(x,y,z)$ factors on $F_2[x,y,z]$ [8].}

%\subsection{REMARKS}
In theorem  \textcolor{black}{\ref{thm:1mod4}}, the fact that $(\phi_{2^k+1},\phi_d)=1$ is a necessary
condition for $f(x)$ not to be exceptional APN. There are many cases
when $\phi_{2^k+1}$ \textcolor{black}{and} $\phi_d$ are relatively
\textcolor{black}{prime}, for example, as we commented, when
$\phi_d$ is absolutely irreducible. In [9], Janwa and Wilson proved,
using different methods including Hensel's lemma implemented on a
computer, that $\phi_d(x,y,z)$ is absolutely irreducible for $3< d
<100$, provided that $d$ is not a Gold or a Kasami-Welch number (in
such cases we know that it reduces). It is easy to show that the
irreducibility of $\phi_d(x,y,z)$ over $\mathbb{F}_2$ also implies
that $(\phi_{2^k+1},\phi_d)=1$. Using all this and the previous
theorems, we get new infinite families of Gold degree
\textcolor{black}{polynomial functions}
that are not exceptional APN.

\section*{Acknowledgement} {The authors are thankful to R.M. Wilson, B. Mishra, H.F. Mattson, Jr., and  F. Castro 
for helpful discussions.}

\end{document}